\documentclass[11pt]{article}

\usepackage[hmargin=1in,vmargin=1.25in]{geometry}
\usepackage{amsmath,amssymb,latexsym,amsthm}
\usepackage{url}

\newtheorem{theorem}{Theorem}
\newtheorem{lemma}{Lemma}
\newtheorem{result}{Result}

\begin{document}

\title{Generalization of a few results in Integer Partitions}
\author{Manosij Ghosh Dastidar\thanks{Ramakrishna Mission Vidyamandira, Belur, West Bengal, India. Email: gdmanosij@gmail.com.} \and Sourav Sen Gupta\thanks{Corresponding author. Indian Statistical Institute, Kolkata, India. Email: sg.sourav@gmail.com.}}
\date{}

\maketitle

\begin{abstract}
In this paper, we generalize a few important results in Integer Partitions; namely the results known as Stanley's theorem and Elder's theorem, and the congruence results proposed by Ramanujan for the partition function. We generalize the results of Stanley and Elder from a fixed integer to an array of subsequent integers, and propose an analogue of Ramanujan's congruence relations for the `number of parts' function instead of the partition function. We also deduce the generating function for the `number of parts', and relate the technical results with their graphical interpretations through a novel use of the Ferrer's diagrams.
\end{abstract}

{\flushleft \textbf{Keywords:}} Stanley's theorem; Elder's theorem; Ramanujan congruences; Ferrer's diagram.

\section{Introduction}
Partitioning a positive integer $n$ as sum of certain positive integers is a well known problem in the domain of number theory and combinatorics. A {\em partition} of a positive integer $n$ is any non-increasing sequence of positive integers that add up to $n$. The partition function $P(n)$ is defined as the number of unordered partitions of $n$. We also define $Q_k(n)$ as the number of occurrences of the part $k$ in all partitions of $n$, $V_k(n)$ as the number of parts occurring $k$ or more times in the partitions of $n$, and $S(n)$ as the sum of the numbers of distinct members in the partitions of $n$. This notation will be followed throughout the paper.

One of the very well referred results in integer partitions is the one presented by Stanley~\cite{stanley1}, which states the following.
\begin{result}[Stanley]
\label{stanley}
The total number of 1's that occur among all unordered partitions of a positive integer is equal to the sum of the numbers of distinct members of those partitions. In terms of the notation, $S(n) = Q_1(n)$.
\end{result}
One direction of generalizing Result~\ref{stanley} is the Elder's theorem~\cite{elder1}, which states the following.
\begin{result}[Elder]
\label{elder}
Total number of occurrences of an integer $k$ among all unordered partitions of $n$ is equal to the number of occasions that a part occurs $k$ or more times in a partition. In terms of the notation, $V_k(n) = Q_k(n)$.
\end{result} 

In this paper, we generalize Result~\ref{stanley} in a different direction than what has been proposed in Result~\ref{elder}. We consider not only a single integer $n$, but generalize the premise to include subsequent integers. Our first result is as follows.
\begin{theorem}
\label{stanleyext}
Given any positive integer $n$ and any positive integer $k$,
$$S(n) \: = \: Q_k(n) + Q_k(n+1) + Q_k(n+2) + \cdots + Q_k(n+k-1) \: = \: 
\sum_{i=0}^{k-1} Q_k(n+i).$$
\end{theorem}
We also generalize Result~\ref{elder} in a similar direction by including subsequent integers into the domain. The formal result is stated as follows.
\begin{theorem}
\label{elderext}
Given any positive integer $n$ and any positive integer $k$,
$$V_k(n) \: = \: Q_{rk}(n) + Q_{rk}(n+k) + Q_{rk}(n+2k) + \cdots + Q_{rk}(n+(r-1)k) \: = \: 
\sum_{i=0}^{r-1} Q_{rk}(n+ik),$$
where $r$ can be chosen to be any positive integer.
\end{theorem}
These two results complete Results~\ref{stanley} and~\ref{elder}, and trace all possible avenues for generalizing the results proposed by Stanley and Elder. We prove both the generalizations in Section~\ref{proofgen}.

In the theory of integer partitions, an array of elegant congruence relations for partition function $P(n)$ were proposed by Ramanujan. He proposed and proved the following.
\begin{result}[Ramanujan]
\label{ramanujan}
For every non-negative $n \in \mathbb{Z}$,
\begin{eqnarray*}
p(5n + 4) & \equiv & 0 \pmod{5},\\
p(7n + 5) & \equiv & 0 \pmod{7},\\
p(11n + 6) & \equiv & 0 \pmod{11}.
\end{eqnarray*}
\end{result}
Ramanujan also conjectured that there exist such congruence modulo arbitrary powers of 5, 7, 11. A lot of eminent mathematicians have worked on similar results for a long time, and the best result till date is: {\em ``there exist such congruence relations for all non-negative integers which are co-prime to 6''}. This result was proved by Ahlgren and Ono~\cite{ono1}. In this paper, we propose a simple analogue to the Ramanujan results that holds true for the function $Q_k(n)$, where $k, n \in \mathbb{Z}$. The formal statement of our analogue is as follows.

\begin{theorem}
\label{ramanujanext}
Given any non-negative integer $n$, following the notation as before, one has
\begin{eqnarray*}
Q_{5} (5n + 4) & \equiv & 0 \pmod{5},\\
Q_{7} (7n + 5) & \equiv & 0 \pmod{7},\\
Q_{11} (11n + 6) & \equiv & 0 \pmod{11}.
\end{eqnarray*}
\end{theorem}

Two common tools for handling integer partitions are Generating functions and Ferrer's diagrams. In the process of generalizing the results of Stanley, Elder and Ramanujan, we also deduce the generating function for $Q_k(n)$ and propose an intuitive explanation of {\em `adding points'} to Ferrer's diagram, which integrates the technical results with their graphical interpretations.

\section{Proof of the Generalizations}
\label{proofgen}
To prove the generalizations stated earlier, we shall require a few preliminary results. One may find the following result in the current literature~\cite{sloane1} on integer partitions.
\begin{result}
\label{res1}
Given any positive integer $n$, one has $Q_1(n) = \sum_{i=0}^{n-1} P(i)$.
\end{result}

We also use the following lemma for our proofs of the generalizations.
\begin{lemma}
\label{lem1}
Given any two positive integers $k, n$, one has $Q_k(n) = Q_k(n-k) + P(n-k)$.
\end{lemma}
\begin{proof}
For a fixed positive integer $k$, a part of size $k$ occurs at least once in all partitions of $n$ of the form $\{k, R\}$, where $R$ denotes a partition of $n-k$. This amounts to at least $P(n-k)$ occurrences of $k$ in partitions of $n$. Moreover, the part $k$ may occur within the partition $R$ of $n-k$, which contributes $Q_k(n-k)$ to the total number of occurrences of $k$. Adding the two contributions, we get the number of occurrences of $k$ in partitions of $n$ as $Q_k(n) = Q_k(n-k) + P(n-k)$.
\end{proof}

These preliminary results will be used to prove the generalizations proposed in this paper. The formal proofs of the main results are presented in the following sections.

\subsection{Proof of Theorem~\ref{stanleyext}}
We have $S(n) = Q_1(n) = \sum_{i=0}^{n-1} P(i)$ by combining Result~\ref{stanley} (Stanley) and Result~\ref{res1}. Using Lemma~\ref{lem1} and solving the recurrence relation therein, we also obtain another known result~\cite{sloane2} as follows.
$$ Q_k(n) =  P(n-k) + P(n-2k) + P(n-3k) + \cdots . $$

Consider the set of partitions $P_n = \{ P(0), P(1), P(2), \ldots, P(n-1) \}$. The sum over all these partitions is $S(n)$, and given any positive integer $k$, one may distribute $P_n$ over disjoint copies of congruence classes $\{ P(i), P(i+1), \ldots, P(i+k-1) \: | \: i \equiv n \bmod k \}$. Thus, one may deduce that
$$ S(n) \: = \: Q_1(n) \: = \: \sum_{i=0}^{n-1} P(i) \: = \: \sum_{j=0}^{k-1} \left( P(n+j-k) + P(n+j-2k) + \cdots \right) \: = \: \sum_{j=0}^{k-1} Q_k(n+j).$$
Hence the result, which holds true for any positive integral values of $n$ and $k$.

\subsection{Proof of Theorem~\ref{elderext}}
In this case, we start with Result~\ref{elder} (Elder), which states $V_k(n) = Q_k(n)$. From the proof of Theorem~\ref{stanleyext}, we have the representation of $Q_k(n)$ as 
$$ Q_k(n) =  P(n-k) + P(n-2k) + P(n-3k) + \cdots . $$

Consider the set of partitions $Q_n = \{ P(n-k), P(n-2k), P(n-3k), \ldots, P(0) \}$. The sum over all these partitions is $Q_k(n)$, and given any positive integer $r$, one may distribute $Q_n$ over disjoint copies of congruence classes $\{ P(i), P(i+k), \ldots, P(i+(r-1)k) \: | \: i \equiv n \bmod rk \}$. Thus, we get
\begin{eqnarray*}
V_k(n) \: = \: Q_k(n) \: & = & P(n-k) + P(n-2k) + P(n-3k) + \cdots \\ 
& = & \sum_{j=0}^{r-1} \left( P(n+jk-rk) + P(n+jk-2rk) + \cdots \right) \: = \: \sum_{j=0}^{r-1} Q_{rk}(n+jk).
\end{eqnarray*}
Hence the result, which holds true for any positive integral values of $n$, $k$ and $r$.

\subsection{Proof of Theorem~\ref{ramanujanext}}
Let us prove the case for $Q_5(n)$ and the rest will follow in a similar fashion. Note that we have the following representation for $Q_5(5n + 4)$
$$ Q_5(5n + 4) =  P(5n + 4 - 5) + P(5n + 4 - 10) + P(5n + 4 - 15) + \cdots, $$
where each $P(\cdot)$ term in the expansion is of the same form $P(5m + 4)$. Thus, each term on the right hand side satisfy $P(\cdot) \equiv 0 \pmod{5}$ as per Ramanujan's congruence results (Result~\ref{ramanujan}). Hence, in turn, $Q_5(5n+4) \equiv 0 \pmod{5}$ as well.

The same is true for $Q_7(7n+5)$ and $Q_{11}(11n + 6)$. One can also derive analogous results for higher order Ramanujan congruences. For example, 
\begin{eqnarray*}
Q_5(25n+24) & \equiv & 0 \pmod{5^2},\\ 
Q_5(125n+99) & \equiv & 0 \pmod{5^3}.
\end{eqnarray*}
In fact, one may also prove that if there exist integers $A(m)$ and $B(m)$ such that $P(A(m) \cdot n + B(m)) \equiv 0 \pmod{m}$, then it can be proved easily that $Q_{C(m)} (A(m) \cdot n + B(m)) \equiv 0 \pmod{m}$ for some positive integer $C(m)$ that depends on $m$ and $B(m)$.

\section{Other Results}
\label{other}
In this section, we deduce the generating function of $Q_k(n)$ and put forward a graphical understanding of the technical results in terms of Ferrer's diagrams.

\subsection{Generating function of $Q_k(n)$}
As we deal with the function $Q_k(n)$, it is also interesting to study the generating function of this parameter. The generating function for the partition function $P(n)$ is known to be
$$ F(x) = \sum_{m=0}^{\infty} P(m) \cdot x^m  = \prod_{n=1}^{\infty} \frac{1}{1 - x^n}$$ 
where we assume $P(0) = 1$. In this formula, we count the coefficient of $x^m$ on both sides, where the coefficient on the right hand side is the result of counting all possible ways that $x^m$ is generated by multiplying smaller or equal powers of $x$. This obviously gives the number of partitions of $m$ into smaller or equal parts. What we require for $Q_k(n)$ is to count the number of $k$'s occurring in each of these partitions. Thus, we want to (i) add $r$ to the count if $x^{rk}$ is a member involved from the right hand side, and (ii) not count any of the partitions where no power of $x^k$ is involved. This intuition gives rise to the following generating function for $Q_k(n)$.
\begin{eqnarray*}
G_k(x) \: = \: \sum_{m=0}^{\infty} Q_k(m) \cdot x^m  &=& \frac{1 \cdot x^k + 2 \cdot x^{2k} + 3 \cdot x^{3k} + \cdots }{(1 - x) \cdot (1 - x^2) \cdots (1 - x^{k-1}) \cdot (1 - x^{k+1}) \cdots }\\
& = & (1 - x^k) \cdot (x^k + 2 x^{2k} + 3 x^{3k} + \cdots )  \cdot \prod_{n=1}^{\infty} \frac{1}{1 - x^n} \\
& = & (x^k + x^{2k} + x^{3k} + \cdots )  \cdot \prod_{n=1}^{\infty} \frac{1}{1 - x^n} \:  = \: \frac{x^k}{1 - x^k} \cdot \prod_{n=1}^{\infty} \frac{1}{1 - x^n}.
\end{eqnarray*}

\subsection{Adding Points to Existing Partitions}
Ferrer's diagram is a tool to graphically represent the partitions of an integer using linear horizontal array of dots/stars to denote each partition. In this section, we shall propose and prove a new problem in partition theory using the elegant exposition of Ferrer's diagram. 

Before we prove our result for adding points to existing Ferrer's diagram, let us define the norms with an illustrative example. Consider the Ferrer's diagram of all partitions of $5$, as in Figure~\ref{fer5}. 

\begin{figure}[htb]
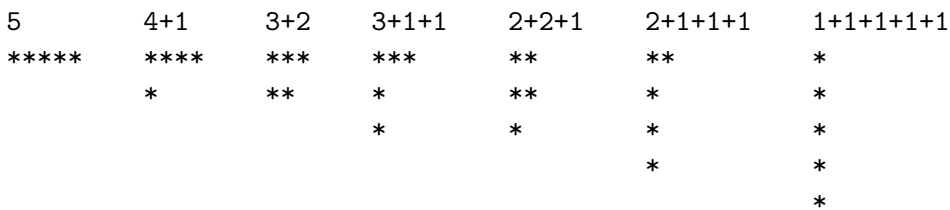

\begin{verbatim}
   5        4+1     3+2    3+1+1    2+2+1    2+1+1+1    1+1+1+1+1
   *****    ****    ***    ***      **       **         *
            *       **     *        **       *          * 
                           *        *        *          *
                                             *          *
                                                        * 
\end{verbatim}
\vspace*{-15pt}
\caption{Ferrer's Diagram for all Partitions of 5}
\label{fer5}
\end{figure}

Let us add one new point to each of the diagrams in this figure such that the resulting arrangements also correspond to valid Ferrer's diagrams. One way is to put the new point as a completely distinct partition in each of the existing ones, and thus get valid Ferrer's diagrams as output (\# denotes the new point in the figure). Another valid way to add the new point is to add it to the existing partitions instead of taking it as a new part. All possibilities are shown in Figure~\ref{fer51}.

\begin{figure}[htb]
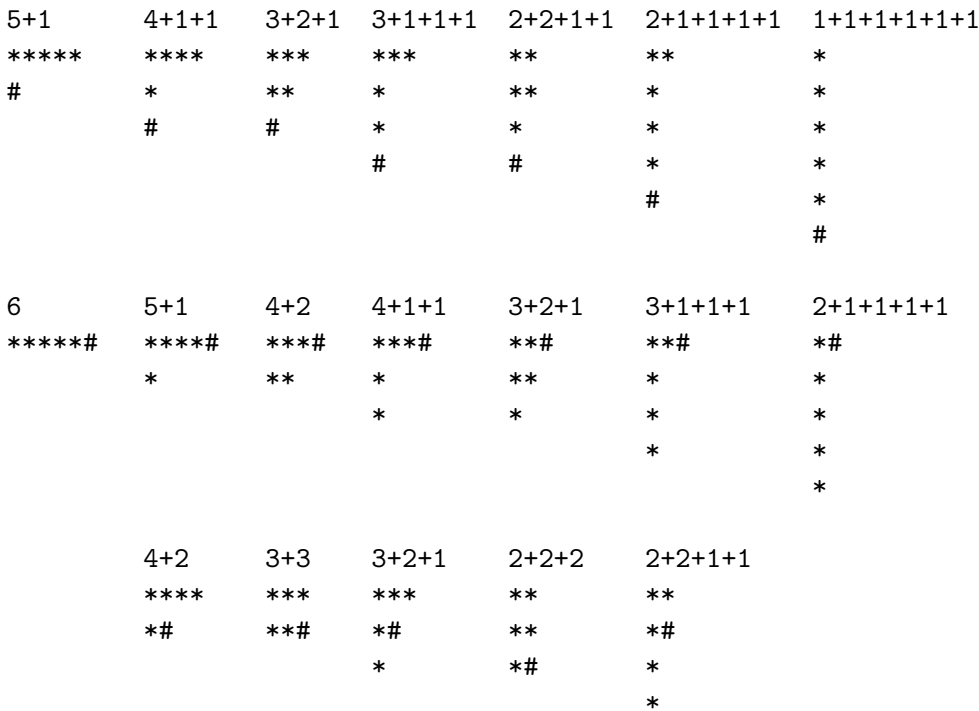

\begin{verbatim}
   5+1      4+1+1   3+2+1  3+1+1+1  2+2+1+1  2+1+1+1+1  1+1+1+1+1+1
   *****    ****    ***    ***      **       **         *
   #        *       **     *        **       *          *
            #       #      *        *        *          *
                           #        #        *          *
                                             #          *
                                                        #

   6        5+1     4+2    4+1+1    3+2+1    3+1+1+1    2+1+1+1+1
   *****#   ****#   ***#   ***#     **#      **#        *#
            *       **     *        **       *          * 
                           *        *        *          *
                                             *          *
                                                        * 

            4+2     3+3    3+2+1    2+2+2    2+2+1+1 
            ****    ***    ***      **       **        
            *#      **#    *#       **       *#         
                           *        *#       *         
                                             *         
\end{verbatim}
\vspace*{-15pt}
\caption{Adding one point to Ferrer's diagram of 5.}
\label{fer51}
\end{figure}

Next, let us explain the process of adding more than one point to a specific diagram. Consider the partition $2+2+1$ of $5$, as shown in Figure~\ref{fer5}. To add 2 points to this partition, say, we consider the addition of the new points as a `packet of 2', instead of adding two separate points. Moreover, we only add this packet (containing two points) in its vertical orientation, i.e, in the form $1+1$, as shown in Figure~\ref{fer221}. With this restrictions imposed on the addition of 2 points, the new partitions that can be generated from $2+2+1$ are as illustrated in Figure~\ref{fer221}.

\begin{figure}[htb]
\begin{verbatim}
   2+2+1    Wrong    Wrong     Right     2+2+1+1+1   3+3+1   Wrong

   **       ##        #        #         **          **#     **
   **                #         #         **          **#     **
   *                                     *           *       *#
                                         #                    #
                                         #
\end{verbatim}
\vspace*{-15pt}
\caption{Adding two new points to partition $2+2+1$.}
\label{fer221}
\end{figure}

Note that we do not allow the packet to be added in horizontal or diagonal orientation, and we also abide by the norms of Ferrer's diagram while adding the packet. Based on this notion of point addition to existing Ferrer's diagram, let us propose and prove the following general result.

\begin{theorem}
\label{ferrer}
Consider adding $k$ points to all partitions of a positive integer $n$ in the Ferrer's diagram, where the new $k$ points are added as a single packet with vertical orientation, as discussed before. Then the total number of new partitions generated in this fashion is equal to the total number of $k$'s occurring in all the partitions of $n+k$.
\end{theorem}
\begin{proof}
While adding the packet of $k$ points to the partitions of $n$, we will count the new partitions in terms of the categorization we made before; adding the packet as a separate unit, and merging the packet with existing parts. It is quite clear that if we add the packet of $k$ points as a separate unit, then each partition of $n$ will generate a single new partition, namely, the existing partition plus $1 + 1 + \cdots + 1$ ($k$ number of 1's). Thus, the total number of partitions generated in this fashion is $P(n)$, the total number of partitions of $n$.

On the other hand, if we look for merging the new points with existing parts, we can fit in the vertical packet of $k$ stars in the Ferrer's diagram if and only if there is a vertical `permissible' opening, i.e., if there is a vertical slot of length $k$ (or more) where one can put this packet without violating the construction rules of Ferrer's diagram. This is possible when there exist at least $k$ copies of the same part in the existing partition. Two cases arise in such a merging situation:
\begin{itemize}
\item If there are $k$ equal parts in a partition, we will have just enough space to fit $k$ vertical points.
\item If there are more than $k$ equal parts, we will still be able to fit just one packet of $k$ points.
\end{itemize}
Thus, the number of new partitions that will be generated in this fashion from an existing partition is the number of parts that occur $k$ times or more in the existing partition. Note that this count is precisely the one mentioned in Elder's theorem (Result~\ref{elder}), i.e., $V_k(n)$.

Considering both possible categories of adding $k$ points to the partitions of $n$, we get the cumulative count of new partitions as $P(n) + V_k(n)$. We further obtain
$$ P(n) + V_k(n) = P(n) + Q_k(n) = Q_k(n+k) $$
from Result~\ref{elder} and Lemma~\ref{lem1}. Hence the desired result.
\end{proof}

Theorem~\ref{ferrer} provides a nice combinatorial intuition towards the problem of adding points to an existing partition, and also integrates Elder's theorem with the extension of Stanley's theorem and the related Lemma. Further explorations in this direction would be to study the general construction of larger partitions using smaller ones as building blocks.

\section{Conclusion}
\label{conclusion}
In this paper, we generalize Stanley's theorem, Elder's theorem in Integer Partitions by including the notion of subsequent integers in each of the original results. The original results were based on a fixed integer $n$ while we generalize it to include the set of all integers $\{n, n+1, n+2, \ldots \}$ in a natural way. Moreover, we propose analogues of Ramanujan's congruence results for the `number of parts' function $Q_k(n)$ instead of the original presentation for the `partition function' $P(n)$. We show that it is natural to extend all Ramanujan-like congruence relations to $Q_k(n)$ from $P(n)$. In this process of studying $Q_k(n)$, we also deduce the generating function for $Q_k(n)$, and relate the technical results with their graphical interpretations through a novel use of the Ferrer's diagrams.

\end{document}